\documentclass[reqno]{amsart}
\usepackage{amsmath, amssymb}

\newtheorem{theo}{Theorem}
\newtheorem{corr}{Corollary}

\newtheorem{rem}{Remark}

\newcommand{\abs}[1]{\ensuremath{\left|#1\right|}}

\def \de{\partial}
\def \om{\omega}
\def \Om{\Omega}
\def \Ombar{\overline{\Omega}}
\def \rtri{\mathbb{R}^3}

\def \ep{\varepsilon}
\def \z{\mathbf{z}}
\def \ve{\mathbf{v}}

\def \T{\mathbf{T}}

\def \te{\mathcal{T}}
\def \sil{\rightarrow}
\def \O{\mathcal{O}}
\def \bO{\boldsymbol{\mathcal{O}}}

\def \G{\mathcal{G}}
\def \bG{\boldsymbol{\mathcal{G}}}
\def \u{\mathbf{u}}
\def \w{\mathbf{w}}

\def \d{\mathrm{d}}
\def \n{\mathbf{n}}
\def \e{\mathbf{e}}
\def \x{\mathbf{x}}
\def \y{\mathbf{y}}
\def \X{\mathbf{X}}
\def \f{\mathbf{f}}
\def \F{\mathbf{F}}
\def \U{\mathbf{U}}

\def \H{\mathbf{H}}
\def \D{\mathbf{D}}
\def \A{\mathbf{A}}
\def \B{\mathbf{B}}
\def \C{\mathbf{C}}
\def \Z{\mathbf{Z}}
\def \E{\mathcal{E}}
\def \DD{\mathcal{D}}
\def \BB{\mathcal{B}}
\def \Re{\mathcal{R}}
\def \We{\mathcal{W}}
\def \M{\mathcal{M}}
\def \bnul{\mathbf{0}}
\def \aone{\alpha_1}

\renewcommand{\div}{{\rm div}\,}

\def \sq{\unskip\nobreak\kern5pt\nobreak\vrule height4pt width4pt depth0pt} 

\begin{document}

\title{On the 3D steady flow of a second grade fluid past an obstacle}
\author{Pawe\l{} Konieczny}
\address{Institute for Mathematics and Its Applications, University of Minnesota\\
114 Lind Hall, 207 Church Street SE, Minneapolis, MN 55455, USA\\
E-mail: konieczny@ima.umn.edu}

\author{Ond\v rej Kreml}
\address{Mathematical Institute of Charles University\\
Sokolovsk\'a 83, 186 75 Praha 8, Czech Republic\\
E-mail: kreml@karlin.mff.cuni.cz}

\thanks{The work of the first author was partially supported by MS grant No. N N201 547438. The second author was supported by the project LC06052 (Jind\v{r}ich Ne\v{c}as Center for Mathematical Modeling), by the Charles University Grant Agency under Contract 2509/2007 and by the Grant Agency of the Czech Republic (number GA201/08/0315) }

\keywords{Viscoelastic fluid, polymeric fluid, Oldroyd model, local strong solution, $L^p$ theory}
\subjclass{35Q35, 76D03, 35Q30}

\begin{abstract}
We study steady flow of a second grade fluid past an obstacle in three space dimensions. We prove existence of solution in weighted Lebesgue spaces with anisotropic weights and thus existence of the wake region behind the obstacle. We use properties of the fundamental Oseen tensor together with results achieved in \cite{Koch} and properties of solutions to steady transport equation to get up to arbitrarily small $\ep$ the same decay as the Oseen fundamental solution.
\end{abstract}

\maketitle

\section{Introduction}
The flow of a second grade fluid is governed by the following system of equations.
\begin{equation}\label{1}
\left.
\begin{aligned}
\displaystyle\rho\frac{\partial\textbf{v}}{\partial t} + \rho\textbf{v}\cdot\nabla\textbf{v} + \nabla p = \text{div }\textbf{T} + \rho\f \\
\text{div } \textbf{v} = 0\\
\end{aligned}
\right\}
\text{in } (0,T) \times \Omega,
\end{equation}
where $\ve$ denotes the fluid velocity, $p$ is the pressure, $\rho$ is the constant density of the fluid, $\f$ stands for the external force and $\T$ is the Cauchy stress tensor which for the second grade fluid is given by (see e.g. \cite{TrNo})
\begin{equation}\label{2} 
\textbf{T} = 2\mu\D + 2\alpha_1\A_1 + 4\alpha_2\D^2.
\end{equation}
Here $\mu$ is a constant viscosity, $\D = \frac{1}{2}(\nabla\ve + (\nabla\ve)^T)$ is the symmetric part of the velocity gradient, $\alpha_1 > 0$ and $\alpha_2$ are the stress moduli and $\A_1$ is given by
\begin{equation}\label{3}
\A_1 = \frac{\de}{\de t}\D + \ve\cdot\nabla\D + (\nabla\ve)^T\D + \D\nabla\ve.
\end{equation}
The condition of thermodynamical stability yields $\alpha_1 + \alpha_2 = 0$, see \cite{DuFo}.
\begin{rem}
The question of signs and values of the stress moduli $\alpha_1$, $\alpha_2$ and especially of $\alpha_1 + \alpha_2$ in this model is not clear. In \cite{CoGa} the authors show that the constraint $\alpha_1 + \alpha_2 = 0$ is not necessary for the mathematical problem being well set. In \cite{GaPaRa} the authors show that for $\alpha_1 < 0$ the rest state of flow of second grade fluid in exterior domain is instable. Our results can be easily adapted also for the case $\alpha_1 + \alpha_2 \neq 0$, however we keep this thermodynamic constraint for simplicity. 
\end{rem}
We consider a steady flow past an obstacle, that is $\Om$ is an exterior domain $\rtri \setminus \BB$, where $\BB$ is a simply connected compact set and we assume that $B_{\kappa L}(\bnul) \subset \BB \subset B_L(\bnul)$ for some $\kappa > 0$ and $L > 0$. Plugging (\ref{1}) - (\ref{3}) together we get
\begin{equation}\label{4}
\begin{split}
\left.
\begin{aligned}
\displaystyle -\mu\Delta\ve -\aone(\ve\cdot\nabla)\Delta\ve + \nabla p = -\rho(\ve\cdot\nabla)\ve + \rho\f + \quad\\
+ \aone\div[(\nabla\ve)^T(\nabla\ve + (\nabla\ve)^T)]\quad\\
\div \ve = 0\quad\\
\end{aligned}
\right\}
\text{in }& \Omega \\
\ve = \bnul \qquad\text{ on }& \de\Om = \de\BB \\
\ve \sil \ve_\infty \qquad\text{ as }& \abs{\x} \sil \infty,
\end{split}
\end{equation}
where $\ve_\infty$ is the prescribed constant velocity at infinity. Assuming $\ve_\infty \neq \bnul$ we can rotate the coordinate system in such a way that $\ve_\infty = \beta\e_1 = (\beta,0,0)$ and denoting $\u = \ve - \ve_\infty$ we get from (\ref{4})
\begin{equation}\label{5}
\begin{split}
\left.
\begin{aligned}
\displaystyle -\mu\Delta\u -\aone(\u\cdot\nabla)\Delta\u -\aone\beta\Delta\frac{\de\u}{\de x_1} +\rho\beta\frac{\de\u}{\de x_1} + \nabla p = \quad\\ -\rho(\u\cdot\nabla)\u + \rho\f + \aone\div[(\nabla\u)^T(\nabla\u + (\nabla\u)^T)]\quad\\
\div \u = 0\quad\\
\end{aligned}
\right\}
\text{in }& \Omega \\
\u = -\ve_\infty = -\beta\e_1 \qquad\text{ on }& \de\Om = \de\BB \\
\u \sil \bnul \qquad\text{ as }& \abs{\x} \sil \infty.
\end{split}
\end{equation}
Next we rewrite the equations in dimensionless form, i.e. we introduce new velocity $\U = \u / \beta$ and new independent variable $\X = \x / L$, where $L$ is the diameter of the obstacle. We renormalize the pressure $P = \frac{p}{\rho\beta^2}$ and the external force $\F = \frac{\f L}{\beta^2}$ and introduce the Reynolds number $\Re = \frac{\rho\beta L}{\mu}$ and the Weisenberg number $\We = \frac{\aone\beta}{L \mu}$. However, for the sake of transparency, we keep writing small letters instead of capital letters. After renormalization we end up with
\begin{equation}\label{6}
\begin{split}
\left.
\begin{aligned}
\displaystyle -\Delta\u -\We(\u\cdot\nabla)\Delta\u -\We\Delta\frac{\de\u}{\de x_1} +\Re\frac{\de\u}{\de x_1} + \Re\nabla p = \quad\\ -\Re(\u\cdot\nabla)\u + \Re\f + \We\div[(\nabla\u)^T(\nabla\u + (\nabla\u)^T)]\quad\\
\div \u = 0\quad\\
\end{aligned}
\right\}
\text{in }& \Om \\
\u = -\e_1 \quad\text{ on }& \de\Om \\
\u \sil \bnul \qquad\text{ as }& \abs{\x} \sil \infty,
\end{split}
\end{equation}
where the renormalized domain $\Om = \rtri \setminus \DD$ and $B_\kappa(\bnul) \subset \DD \subset B_1(\bnul)$. Finally we follow the decomposition procedure proposed in \cite{MoSo}. We introduce new pressure $q$ as a solution to
\begin{equation}\label{Pr1}
q + \We[(\u + \e_1)\cdot\nabla] q = \Re p
\end{equation}
and we denote
\begin{equation}\label{Pr2}
-\Delta\u + \Re\frac{\de\u}{\de x_1} + \nabla q =: \z.
\end{equation}
Then $\z$ satisfies
\begin{multline}\label{Pr3}
\z + \We[(\u + \e_1)\cdot\nabla] \z = \Re\f - \Re(\u\cdot\nabla)\u + \We\div[(\nabla\u)^T(\nabla\u + (\nabla\u)^T)] - \\
 - \We(\nabla\u)^T\nabla q + \Re\We(\u\cdot\nabla)\frac{\de\u}{\de x_1} + \Re\We\frac{\de^2\u}{\de x_1^2}.
\end{multline}
Note that we still have the conditions
\begin{equation}\label{Pr4}
\begin{split}
\div \u &= 0 \qquad\text{ in } \Om\\
\u &= -\e_1 \quad\text{ on } \de\Om \\
\u &\sil \bnul \qquad\text{ as } \abs{\x} \sil \infty,
\end{split}
\end{equation}

Our main result is the following
\begin{theo}\label{maintheo}
Let $\f = \div \H$, $\H \in W^{k,2}(\Om)$, $k \geq 3$. Let $\Om \in C^{k+1}$ be an exterior domain in $\rtri$ and let $\Re_0, \We_0$ be sufficiently small. Then for any $\Re \in (0,\Re_0), \We \in (0,\We_0)$ there exists a unique solution $(\u,q)$ to the problem (\ref{Pr1}) - (\ref{Pr4}) for which the following estimates hold
\begin{equation}
\Re^{\frac{1}{4}}\left\|\u\right\|_4 + \left\|\nabla\u\right\|_{k,2} + \left\|q\right\|_{k,2} \leq K.
\end{equation}
If in addition $\f,\H \in L^p(\Om,\mu_{1-\frac{2}{p}}^{\frac{3}{2}-\frac{3}{p},2\om}(\cdot,\Re))$ for some $p>6$ and $\Re$ and $\We$ are sufficiently small, the previously obtained solution $(\u,q)$ has the following properties
\begin{equation}
\begin{split}
\u &\in L^p(\Om,\mu_{1-\frac{2}{p}}^{1-\frac{3}{p},\om}(\cdot,\Re)) \\
\nabla\u, \nabla^2\u &\in L^p(\Om,\mu_{1-\frac{2}{p}}^{\frac{3}{2}-\frac{3}{p},\om}(\cdot,\Re)) \\
q, \nabla q &\in L^p(\Om,\mu_{\frac{1}{2}-\frac{2}{p}}^{1-\frac{3}{p},\om}(\cdot,\Re)).
\end{split}
\end{equation}
In particular 
\begin{equation}
\u \in L^\infty(\Om,\mu_{1-\frac{2}{p}}^{1-\frac{3}{p},\om}(\cdot,\Re)).
\end{equation}
\end{theo}
\begin{rem}
The weights $\mu_B^{A,\om}$ are defined in (\ref{weights}). As the power $p$ can be chosen arbitrarily large, we get almost the same asymptotic structure as for the fundamental solution $\bO$ of the Oseen system.
\end{rem}
\begin{rem}
Note that for $\f \in L^\infty(\Om,\eta^{\frac{3}{2}}_1(\x))$ it holds
\begin{equation}
\left\|\f\right\|_{L^p(\Om,\mu^{\frac{3}{2}-\frac{3}{p},2\om}_{1-\frac{2}{p}}(\cdot,\Re))} \leq C\Re^{-2\om-\frac{3}{p}}\left\|\f\right\|_{L^\infty(\Om,\eta^{\frac{3}{2}}_1(\cdot))}.
\end{equation}
\end{rem}

Throughout this paper we shall assume $\We$ and $\Re$ small. We introduce operator 
\begin{equation}\label{mapping}
\M : (\w,s) \mapsto \z \mapsto (\u,q),
\end{equation}
where for given $(\w,s)$, $\z$ is the solution to the transport equation
\begin{multline}\label{transport}
\z + \We[(\w + \e_1)\cdot\nabla] \z = \Re\f - \Re(\w\cdot\nabla)\w + \We\div[(\nabla\w)^T(\nabla\w + (\nabla\w)^T)] - \\
 - \We(\nabla\w)^T\nabla s + \Re\We(\w\cdot\nabla)\frac{\de\w}{\de x_1} + \Re\We\frac{\de^2\w}{\de x_1^2} =: \B(\f,\w,s) \qquad \text{in } \Om
\end{multline}
and $(\u,q)$ is the solution to the Oseen problem
\begin{equation}\label{Oseen}
\begin{split}
-\Delta\u + \Re\frac{\de\u}{\de x_1} + \nabla q &= \z \qquad\text{ in } \Om\\
\div \u &= 0 \qquad\text{ in } \Om \\
\u &= -\e_1 \quad\text{ on } \de\Om \\
\u &\sil \bnul \qquad\text{ as } \abs{\x} \sil \infty.
\end{split}
\end{equation}
We have decomposed the original problem into the Oseen problem (\ref{Oseen}) and the steady transport equation (\ref{transport}). Due to the pressence of the Oseen problem we expect the structure of solutions to correspond to the structure of the Oseen fundamental solution, especially the existence of the wake region behind the obstacle (compare with \cite{Finn}, \cite{Smith} for incompressible Navier-Stokes equations and \cite{NoPo} for viscoelastic fluid). Denoting $s(\x) = \abs{\x} - x_1$ one might expect the solution $\u$ to satisfy
\begin{equation}
\abs{\u(\x)} \leq C\abs{\x}^{-1}(1 + s(\x))^{-1}
\end{equation}
for $\abs{\x}$ sufficiently large. However we are only able to prove
\begin{equation}
\abs{\u(\x)} \leq C\abs{\x}^{-1+\ep}(1 + s(\x))^{-1+\ep}
\end{equation}
for arbitrarily small $\ep$. This is due to the presence of the linear term $\Re\We\frac{\de^2\u}{\de x_1^2}$ on the right-hand side of the transport equation (\ref{transport}). This implies that the solution to the transport equation has the same decay as $\nabla^2\w$ (quadratic terms decay faster). Moreover, for the Oseen system previously available estimates for the second gradient (which were using techniques of fundamental solution) lose logarithmic factor in the weight in the $L^\infty$ norm, $L^p$ estimates lose $\ep$ in the weight, see \cite{PoPhd}. Thus fixed point theorem argument would not work. Fortunately due to recent results of Koch \cite{Koch} we have at least $L^p$ estimates without mentioned $\ep$ loss in the weight and therefore fixed point argument works.

It is worth mentioning that in \cite{NoPo}, where the model of viscoelastic fluid is considered, authors are able to overcome these problems by introducing modified Oseen problem with the problematic term $\frac{\de^2\u}{\de x_1}$ being included in the Oseen operator. Then all terms on the right hand side are quadratic and thus with better decay. In our problem this cannot be repeated since this linear term appears in the transport equation.

The drawback of using $L^p$ estimates is that in order to get $L^\infty$ estimate we have to use embedding theorems and thus we are able to prove 
\begin{equation*}
\abs{\u(\x)} \leq C\abs{\x}^{-1+\ep}(1 + s(\x))^{-1+\ep}
\end{equation*}
for arbitrarily small $\ep$. Details will be specified in the proper part of the proof.

\section{Preliminaries}
Throughout this paper we will use standard notation for the Lebesgue spaces $L^p(\Om)$ with the norm $\left\|\cdot\right\|_p$, the Sobolev spaces $W^{k,p}(\Om)$ with the norm $\left\|\cdot\right\|_{k,p}$ and the homogenous Sobolev spaces $D^{k,p}(\Om)$ with the norm $\abs{\cdot}_{k,p}$. Let $g \in L^1_{loc}(\Om)$ be a nonnegative weight. Then $L^p(\Om,g)$ denotes the weighted $L^p$ space with the norm
\begin{equation*}
\left\|u\right\|_{p,(g)} = \left\|ug\right\|_p
\end{equation*}
for any $p \in [1,\infty]$. Similarly, $W^{k,p}(\Om,g)$ denotes the weighted Sobolev space with the norm 
\begin{equation*}
\left\|u\right\|_{k,p,(g)} = \left\|ug\right\|_{k,p}.
\end{equation*}
Note that if there is no confusion we sometimes omit writing the domain and instead of $L^p(\Om,g)$ we write simply $L^p(g)$.

As we decomposed the original problem into an Oseen problem and a steady transport equation, we shall mention several classical results about these problems in three dimensional exterior domains. Let us start with the Oseen problem (\ref{Oseen}).

\subsection{Oseen problem. Existence and properties}

We denote by $(\bO,\e)$ the fundamental solution to the Oseen problem. It can be shown (see for example \cite{PoPhd}) that
\begin{equation}
\e(\x) = \nabla\E(\x),
\end{equation}
where $\E(\x)$ is the fundamental solution to the Laplace equation. The tensor $\bO(\x,\Re)$ (here $\Re$ denotes the constant standing in front of $\frac{\de\u}{\de x_1}$ in the equation) satisfies the following property
\begin{equation}
\bO(\x,\Re) = \Re\bO(\Re\x,1)
\end{equation}
and therefore it is sufficient to study the tensor $\bO(\x,1)$. For $\abs{\x} \sil \infty$ we have
\begin{equation}
\begin{split}
\bO(\x,1) &\sim \abs{\x}^{-1}(1 + s(\x))^{-1} \\
D^\alpha\bO(\x,1) &\sim \abs{\x}^{-1-\frac{\abs{\alpha}}{2}}(1 + s(\x))^{-1-\frac{\abs{\alpha}}{2}} \\
D^\alpha\bO(\x,1) &\sim \abs{\x}^{-1-\alpha_1-\frac{\abs{\alpha}-\alpha_1}{2}}(1 + s(\x))^{-1-\frac{\abs{\alpha}-\alpha_1}{2}},
\end{split}
\end{equation}
i.e. the derivatives with respect to the first variable decay faster.

Next, we present some results for the general Oseen problem
\begin{equation}\label{basicOseen}
\begin{split}
\displaystyle -\Delta\u + \Re\frac{\de\u}{\de x_1} + \nabla P &= \f = \div \bG \quad\text{ in } \Om\\
\div \u &= 0 \quad\text{ in } \Om \\
\u &= \u_* \quad\text{ on } \de\Om \\
\u &\sil \bnul \quad\text{ as } \abs{\x} \sil \infty,
\end{split}
\end{equation}
where $\Om$ is an exterior domain.

The proof of the following classical theorem can be found in \cite{Ga} or in \cite{PoPhd}.

\begin{theo}\label{theo1}
Let $\Om \subset \rtri$ be an exterior domain of class $C^{k+2}$. Let $\f \in D_0^{-1,q}(\Om) \cap W^{k,2}(\Om)$, $\u_* \in W^{k+\frac{3}{2},2}(\de\Om)$, $q \in (\frac{3}{2},3)$, $k \geq 0$. Then there exists exactly one q-weak solution (i.e. weak solution such that $\u \in W^{1,q}(\Om)$) to (\ref{basicOseen}). Moreover
\begin{equation}
\u \in L^{\frac{4q}{4-q}}(\Om) \quad \text{and} \quad \nabla\u, P \in L^q(\Om) \cap W^{k+1,2}(\Om)
\end{equation}
and
\begin{multline}
a_2\left\|\u\right\|_\frac{4q}{4-q} + \left|\u\right|_{1,q} + \left\|\nabla\u\right\|_{k+1,2} + \left\|P\right\|_q + \left\|P\right\|_{k+1,2} \leq \\
\leq C(\left|\f\right|_{-1,q} + \left\|\f\right\|_{k,2} + \left\|\u_*\right\|_{k+\frac{3}{2},2,\de\Om}),
\end{multline}
where for $\Re \in (0,\Re_0]$ the constant $C = C(k,q,\Om,\Re_0)$ and $a_2 = \min\{1,\Re^\frac{1}{4}\}$.
\end{theo}


We need the following integral representation of solutions to (\ref{basicOseen}) to obtain weighted estimates.

Let us denote
\begin{equation}\label{tt}
\begin{split}
T_{ij}(\u,P) &= \frac{\de u_i}{\de x_j} + \frac{\de u_j}{\de x_i} - P\delta_{ij} \\
\te_{ij}(\e) &= \frac{\de e_i}{\de x_j} + \frac{\de e_j}{\de x_i} + \Re e_1\delta_{ij}
\end{split}
\end{equation}

\begin{theo}
Let $\Om \in C^2$ be an exterior domain, $\bG \in C_0^\infty(\Ombar)$ and $(\u,P)$ be the unique solution to (\ref{basicOseen}). Let $\T$ be defined in (\ref{tt}) and $(\bO,\e)$ be the fundamental solution to the Oseen problem. Then
\begin{multline}\label{int0}
u_j(\x) = \int_\Om \frac{\de}{\de x_k} \O_{ij}(\x-\y,\Re)\G_{ik}(\y)\d\y + \\
+ \int_{\de\Om}\left[-\Re\O_{ij}(\x-\y,\Re)u_i(\y)\delta_{1k} + u_i(\y)T_{ik}(\bO_{\cdot j},e_j)(\x-\y,\Re) \right. + \\
+ \left.\O_{ij}(\x-\y,\Re)T_{ik}(\u,P)(\y) + \O_{ij}(\x-\y,\Re)\G_{ik}(\y)\right]n_k(\y)\d S
\end{multline}
\begin{multline}\label{int1}
D^\alpha u_j(\x) = - \int_\Om D^\alpha \frac{\de}{\de x_k} \O_{ij}(\x-\y,\Re)\G_{ik}(\y)\d\y + \\
+ \int_{\de\Om}\left[-\Re D^\alpha\O_{ij}(\x-\y,\Re)u_i(\y)\delta_{1k} + u_i(\y)D^\alpha T_{ik}(\bO_{\cdot j},e_j)(\x-\y,\Re) \right. + \\
+ \left.D^\alpha \O_{ij}(\x-\y,\Re)T_{ik}(\u,P)(\y) + D^\alpha\O_{ij}(\x-\y,\Re)\G_{ik}(\y)\right]n_k(\y)\d S
\end{multline}
for $\abs{\alpha} = 1$ and
\begin{multline}\label{int2}
D^\alpha u_j(\x) = v.p.\int_\Om D^\alpha \O_{ij}(\x-\y,\Re)\frac{\de}{\de y_k}\G_{ik}(\y)\d\y + c_{ij\alpha_1\alpha_2}\frac{\de \G_{ik}}{\de x_k}(\x) + \\
+ \int_{\de\Om}\left[-\Re D^\alpha\O_{ij}(\x-\y,\Re)u_i(\y)\delta_{1k} + u_i(\y)D^\alpha T_{ik}(\bO_{\cdot j},e_j)(\x-\y,\Re) \right. + \\
+ \left.D^\alpha \O_{ij}(\x-\y,\Re)T_{ik}(\u,P)(\y)\right]n_k(\y)\d S
\end{multline}
for $\abs{\alpha} = 2$. 
\end{theo}

\begin{rem}
The integral representation formulas hold for much larger classes of functions. For example it holds for a.a. $\x \in \Om$ if $\u \in W^{2,q}_{loc}(\Ombar)$ and $P \in W^{1,q}_{loc}(\Ombar)$ for some $q \in (1,\infty)$ and
\begin{itemize}
\item (\ref{int0}) if $\bG \in L^q(\Om)$ and $\div \bG \in L^r_{loc}(\Ombar)$ for $q \in (1,4)$, $r \in (1,\infty)$,
\item (\ref{int1}) if $\bG \in L^q(\Om)$ and $\div \bG \in L^r_{loc}(\Ombar)$ for $q,r \in (1,\infty)$,
\item (\ref{int2}) if $\div \bG \in L^r_{loc}(\Ombar)$ for $r \in (1,\infty)$.
\end{itemize}
\end{rem}

For the pressure we have also integral representation formulas.
\begin{theo}
Let $\Om \in C^2$ be an exterior domain, $\bG \in C_0^\infty(\Ombar)$ and $(\u,P)$ be the unique solution to (\ref{basicOseen}). Let $T_{ij}$ and $\te_{ij}$ be defined in (\ref{tt}). Then
\begin{multline}\label{intp0}
P(\x) = v.p.\int_\Om \frac{\de}{\de x_k} e_i(\x-\y)\G_{ik}(\y)\d\y + c_{ik}\G_{ik}(\x) + \\
+ \int_{\de\Om}\left[-\Re e_i(\x-\y)u_i(\y)\delta_{1l} + u_i(\y)\te_{il}(\e)(\x-\y) \right. + \\
+ \left.e_i(\x-\y)T_{il}(\u,P)(\y) + e_i(\x-\y)\G_{il}(\y)\right]n_l(\y)\d S
\end{multline}
\begin{multline}\label{intp1}
D^\alpha P(\x) = v.p.\int_\Om D^\alpha e_i(\x-\y)\frac{\de}{\de y_k}\G_{ik}(\y)\d\y + c_{ik}\frac{\de}{\de x_k}\G_{ik}(\x) + \\
+ \int_{\de\Om}\left[-\Re D^\alpha e_i(\x-\y)u_i(\y)\delta_{1l} + u_i(\y)D^\alpha \te_{il}(\e)(\x-\y) \right. + \\
+ \left.D^\alpha e_i(\x-\y)T_{il}(\u,P)(\y)\right]n_l(\y)\d S
\end{multline}
for $\abs{\alpha} = 1$. 
\end{theo}

\begin{rem}
The integral representation formulas for pressure hold also for much larger classes of functions. For example it holds for a.a. $\x \in \Om$ if $\u \in W^{2,q}_{loc}(\Ombar)$ and $P \in W^{1,q}_{loc}(\Ombar)$ for some $q \in (1,\infty)$ and
\begin{itemize}
\item (\ref{intp0}) if $\bG \in L^q(\Om)$ and $\div \bG \in L^r_{loc}(\Ombar)$ for $q,r \in (1,\infty)$,
\item (\ref{intp1}) if $\div \bG \in L^r_{loc}(\Ombar)$ for $r \in (1,\infty)$.
\end{itemize}
\end{rem}

The proof of these representation formulas can be found in \cite{Ga} or in \cite{PoPhd}.

\subsection{Results for weighted spaces}

We introduce the following weights which will be useful in studying the asymptotic structure of solutions
\begin{equation}\label{weights}
\begin{split}
\eta^A_B(\x) &= (1+\abs{\x})^A(1+s(\x))^B \\
\nu^A_B(\x) &= \abs{\x}^A(1+s(\x))^B \\
\mu^{A,\om}_B(\x) &= \eta^{A-\om}_B(\x)\nu^\om_0(\x) \\
\eta^A_B(\x,\Re) &= (1+\abs{\Re\x})^A(1+s(\Re\x))^B \\
\nu^A_B(\x,\Re) &= \abs{\x}^A(1+s(\Re\x))^B \\
\mu^{A,\om}_B(\x,\Re) &= \eta^{A-\om}_B(\x,\Re)\nu^\om_0(\x,\Re) 
\end{split}
\end{equation}

We recall that weight $g$ belongs to the class $A_p$ if there exists a constant $C$ such that
\begin{equation}
\sup_Q \left[\left(\frac{1}{\abs{Q}}\int_Q g^p(\x)\d\x\right)\left(\frac{1}{\abs{Q}}\int_Q g^{-\frac{p}{p-1}}(\x)\d\x\right)^{p-1}\right] \leq C < \infty,
\end{equation}
where the supremum is taken over all cubes $Q \subset \rtri$.

A basic property of the weight $\eta_{-b}^{-a}$ is the following condition
\begin{equation}\label{integrationeta}
\int_{\rtri} \eta_{-b}^{-a}(\x)\d\x < \infty \quad \Leftrightarrow \quad a + \min\{1,b\} > 3.
\end{equation}
For a proof of this and further properties we refer the reader to \cite{PoPhd}.

One of the main tools we use is the following theorem due to Koch \cite{Koch}. We should mention that this estimate is an essential improvement of what has been known about weighted estimates of the solutions to the Oseen system. Without this result it was impossible to obtain weighted estimates for the model of the second grade fluid or for Maxwell and Oldroyd-type fluids, see \cite{PoPhd}.
\begin{theo}\label{Koch}
\textbf{(Koch)} Let $T$ be an integral operator with the kernel $\frac{\de^2}{\de x_i \de x_j}\bO$ on $\rtri$. Then the following estimate hold:
\begin{equation}
\left\|T f\right\|_{p,(g),\rtri} \leq C \left\|f\right\|_{p,(g),\rtri}
\end{equation}
for $p \in (1,\infty)$ and $g = \eta_B^A(\x)$ for $A,B$ satisfying
\begin{equation}\label{AB}
\begin{split}
A, B \in \left(-\frac{1}{p},\frac{2(p-1)}{p}\right) \\
A + B > -\frac{1}{p} \\
2A - B, 2B - A < \frac{2(p-1)+1}{p} 
\end{split}
\end{equation}
\end{theo}
The proof of this theorem can be found in \cite{Koch}. Note that in \cite{Koch} the theorem is formulated only for the case $p=2$, nevertheless the proof is given for general $p \in (1,\infty)$. The same estimate holds also for the case of an exterior domain.

\begin{corr}
For the same operator $T$ as in Theorem \ref{Koch} we have
\begin{equation}
\left\|T f\right\|_{p,(g_1),\Om} \leq C\Re^{\om} \left\|f\right\|_{p,(g_2),\Om},
\end{equation}
where $g_1 = \mu_{B}^{A,\om}(\x,\Re)$, $g_2 = \mu_B^{A,2\om}(\x,\Re)$, $A,B$ satisfying (\ref{AB}) and $\om \in [0,\frac{A}{2})$. 
\end{corr}
\begin{proof}
First we observe that for all $\x \in \Om$ it holds
\begin{equation}
\mu_B^{A,\om}(\x) \leq \eta_B^A(\x) \leq (1+\frac{1}{\kappa})^{\om}\mu_B^{A,\om}(\x),
\end{equation}
where $\kappa$ was introduced earlier by condition $B_\kappa(\bnul) \subset \DD$ and $\Om = \rtri \setminus \DD$. In other words the weights $\eta^A_B(\x)$ and $\mu^{A,\om}_B(\x)$ are equivalent in $\Om$. The presence of the term $\Re^\om$ is an easy consequence of a rescaling argument, because $\nabla^2\bO(\x,\Re) = \Re^3\nabla^2\bO(\Re\x,1)$.
\end{proof}

The proofs of the following theorems can be found for example in \cite{KrNoPo} or in \cite{PoPhd}.

\begin{theo}\label{theo6}
Let $T$ be an integral operator with the kernel $\abs{\nabla\bO}$, $T : f \sil \abs{\nabla\bO} \ast f$ and $p \in (1,\infty)$. Then $T$ is a well defined continuous operator:
\begin{equation}
L^p(\rtri,\eta_B^{A+\frac{1}{2}}(\cdot,\Re)) \mapsto L^p(\rtri,\eta_B^A(\cdot,\Re))
\end{equation}
for $B \in (0,\frac{3}{2}-\frac{3}{2p})$, $A+B > -\frac{1}{p}$, $A < \frac{3}{2} - \frac{2}{p}$, $A - B < \frac{1}{2} - \frac{1}{p}$.
Moreover we have for $A,B$ specified above
\begin{equation}
\left\|\abs{\nabla\bO(\cdot,\Re)}\ast f\right\|_{p,\eta_B^A(\cdot,\Re),\rtri} \leq C\Re^{-1}\left\|f\right\|_{p,\eta_B^{A+\frac{1}{2}}(\cdot,\Re),\rtri}
\end{equation}
\end{theo}
\begin{corr}
For the same operator $T$ as in Theorem \ref{theo6} one has
\begin{equation}
\left\|T f\right\|_{p,(g_1),\Om} \leq C\Re^{-1+\om} \left\|f\right\|_{p,(g_2),\Om},
\end{equation}
where $g_1 = \mu_{B}^{A,\om}(\x,\Re)$, $g_2 = \mu_B^{A,2\om}(\x,\Re)$, $A,B$ satisfy the assertions of Theorem \ref{theo6} and $\om \in [0,\frac{A}{2})$. 
\end{corr}

\begin{theo}\label{theo7}
Let 
\begin{equation}
T f(\x) = \frac{\de}{\de x_i} \int_{\rtri} e_j(\x-\y) f(\y)\d\y, \quad i,j = 1,2,3,
\end{equation}
$f \in C_0^\infty(\rtri)$, $p \in (1,\infty)$ and let $g$ stands for one of weights $\eta_B^A$, $\nu_B^A$, $\mu_B^{A,\om}$. Let $A,B$ be such that $g$ is an $A_p$ weight in $\rtri$. Then $T$ maps $C_0^\infty(\rtri)$ into $L^p(\rtri,g)$ and
\begin{equation}
\left\|T f\right\|_{p,(g),\rtri} \leq C\left\|f\right\|_{p,(g),\rtri}.
\end{equation}
Moreover T can be continuously extended onto $L^p(\rtri,g)$.
\end{theo}
\begin{corr}
The same holds also for the case of an exterior domain $\Om$.
\end{corr}

\begin{theo} $\quad$

\begin{itemize}
\item Let $B \in (-\frac{1}{p},\frac{p-1}{p})$ and $A + B \in (-\frac{3}{p},\frac{3(p-1)}{p})$. Then the weight $\eta_B^A$ is an $A_p$ weight in $\rtri$ for $p \in (1,\infty)$.
\item Let moreover $A \in (-\frac{3}{p},\frac{3(p-1)}{p})$ and $\om \in [0,A]$. Then the weights $\nu_B^A$ and $\mu_B^{A,\om}$ are $A_p$ weights in $\rtri$ for $p \in (1,\infty)$.
\end{itemize}
\end{theo}

\subsection{Transport equation}

Next we consider the steady transport equation
\begin{equation}\label{trans}
z + \w \cdot \nabla z = f \qquad \text{ in } \Om.
\end{equation}
This equation is scalar, nevertheless all theorems below hold also for the vector case. The following theorems are proved in \cite{No} even for more complicated cases.

\begin{theo}\label{transiii}
(i) Let $\Om \in C^{0,1}$ be an exterior domain, $\w \in C^{k-1}(\Om)$, $\w\cdot\n=0$ on $\de\Om$, $\nabla^k\w \in L^3(\Om)$, $f \in W^{k,q}$ for $q \in (1,3)$, $kq > 3$. Then there exists $\alpha > 0$ such that if 
\begin{equation}
\left\|\nabla\w\right\|_{C^{k-2}} + \left\|\nabla^k\w\right\|_3 < \alpha,
\end{equation}
then there exists unique solution $z \in W^{k,q}(\Om)$ to (\ref{trans}) satisfying the estimate
\begin{equation}
\left\|z\right\|_{k,q} \leq C(\alpha)\left\|f\right\|_{k,q}.
\end{equation}

(ii) Let $\Om \in C^{0,1}$ be an exterior domain, $\w \in C^k(\Om)$, $\w\cdot\n=0$ on $\de\Om$, $f \in W^{k,q}$ for $kq > 3$. Then there exists $\alpha > 0$ such that if 
\begin{equation}
\left\|\nabla\w\right\|_{C^{k-1}} < \alpha,
\end{equation}
then there exists unique solution $z \in W^{k,q}(\Om)$ to (\ref{trans}) satisfying the estimate
\begin{equation}
\left\|z\right\|_{k,q} \leq C(\alpha)\left\|f\right\|_{k,q}.
\end{equation}
\end{theo}


\begin{theo}\label{weightedtrans}
Let $\Om,k,q,\w$ and $f$ satisfy the assumptions of Theorem \ref{transiii} (ii). Moreover let $g \in C^k(\Om)$ be a positive weight such that $W^{k,q}(\Om,g) \subset W^{k,q}(\Om)$ and let
\begin{equation}
\left\|\w\cdot\nabla\ln g\right\|_{C^{k-1}} + \left|\w\cdot\nabla\ln g\right|_{k,q}
\end{equation}
be sufficiently small. Let $f \in W^{k,q}(\Om,g)$. Then $z$, the solution to (\ref{trans}), belongs to $W^{k,q}(\Om,g)$ and
\begin{equation}
\left\|z\right\|_{k,q,(g)} \leq C\left\|f\right\|_{k,q,(g)}.
\end{equation}
\end{theo}

\section{The proof of Theorem \ref{maintheo}}

\subsection{Existence of solution}

Here we briefly sketch the method of constructing the solution to the system (\ref{Pr1})-(\ref{Pr4}). It is based on the following version of the Banach fixed point theorem.
\begin{theo}\label{Banach}
Let $X, Y$ be Banach spaces such that $X$ is reflexive and $X \hookrightarrow Y$. Let $H$ be nonempty, closed, convex and bounded subset of $X$ and let $\M: H \mapsto H$ be a mapping such that
\begin{equation}
\left\|\M(u) - \M(v)\right\|_Y \leq \delta\left\|u - v\right\|_Y \quad \forall u, v \in H,
\end{equation}
$\delta \in [0,1)$. Then $\M$ has a unique fixed point in $H$.
\end{theo}

The proof of existence of solutions in Sobolev spaces is based on the method described for example in \cite{PoPhd} and \cite{NoPo}. The solution is obtained as a limit of successive approximations 
\begin{equation}
(\u_{n+1},q_{n+1}) = \M(\u_n,q_n), \qquad n \geq 0,
\end{equation}
where the mapping $\M$ was introduced in (\ref{mapping}). 
\begin{theo}\label{exSobol1}
Let $\Om \in C^{k+1}$ be an exterior domain in $\rtri$ and let $\f = \div \H$, $\H \in W^{k,2}(\Om)$, $k \geq 3$. Let $\Re_0$, $\We_0$ be sufficiently small. Then for any $\Re \in (0,\Re_0)$, $\We \in (0,\We_0)$ there exists $(\u,q)$ a solution to the system (\ref{Pr1})-(\ref{Pr4}) such that $\u \in L^4(\Om)$ and $\nabla\u, q \in W^{k,2}(\Om)$.
\end{theo}
\begin{proof}
We use Theorem \ref{Banach} for the following choice of spaces: $X = V_k$, $Y = V_{k-1}$, where 
\begin{equation}
V_k = \left\{(\u,q): \u \in L^4(\Om), \nabla\u, q \in W^{k,2}(\Om)\right\}
\end{equation}
with the norm
\begin{equation}
\left\|(\u,q)\right\|_{V_k} = \Re^{\frac{1}{4}}\left\|\u\right\|_4 + \left\|\nabla\u\right\|_{k,2} + \left\|q\right\|_{k,2}
\end{equation}
We use Theorems \ref{theo1} and \ref{transiii} for $q = 2$ and the observation that the right-hand side of the transport equation (\ref{Pr3}) (we denote it by $\B(\f,\w,s)$) can be written in the divergence form as 
\begin{equation}
\begin{split}
\B(\f,\w,s) = \div \biggl[\Re\H - \Re \w \otimes \w + \We(\nabla\w)^T(\nabla\w + (\nabla\w)^T) - \biggr.\\
\left.- \We(\nabla\w)^T s + \Re\We \frac{\de\w}{\de x_1} \otimes \w + \Re\We \frac{\de\w}{\de x_1} \otimes \e_1 \right] =: \div\C(\H,\w,s)
\end{split}
\end{equation}
assuming $\f = \div \H$.
As a part of the proof we obtain the following estimates
\begin{equation}\label{ass1}
\begin{split}
\Re^{\frac{1}{4}}\left\|\u_n\right\|_4 + \left\|\nabla\u_n\right\|_{k,2} + \left\|q_n\right\|_{k,2} \leq K \\
\left\|\C(\H,\u_n,q_n)\right\|_{k,2} \leq K
\end{split}
\end{equation}
for all $n \geq 0$, for some constant $K > 0$ depending only on $\Re_0, \We_0$ and $\H$.
\end{proof}


\subsection{Weighted estimates}

In this section we study weighted estimates which are crucial to obtain asymptotic behaviour of the solution. As we have used Theorem \ref{Banach} to prove existence of solutions in Sobolev spaces, it is now sufficient to prove that $\M$ maps sufficiently large balls in proper weighted spaces into themselves. Then choosing $(\u_1,q_1)$ from this ball the solution, as the limit of the sequence $(\u_n,q_n)$, belongs to the same ball. We estimate the sequence in the following space
\begin{equation}
\begin{split}
V = \left\{(\u,q) : \u \in L^p(\Om, \mu_{1-\frac{2}{p}}^{1-\frac{3}{p},\om}(\cdot,\Re)), \nabla\u, \nabla^2\u \in L^p(\Om, \mu_{1-\frac{2}{p}}^{\frac{3}{2}-\frac{3}{p},\om}(\cdot,\Re)),\right. \\
\left.q, \nabla q \in L^p(\Om, \mu_{\frac{1}{2}-\frac{2}{p}}^{\frac{3}{2}-\frac{3}{p},\om}(\cdot,\Re))\right\}
\end{split}
\end{equation}
with the norm
\begin{equation}
\begin{split}
\left\|(\u,q)\right\|_V = \left\|\u\right\|_{p,\mu_{1-\frac{2}{p}}^{1-\frac{3}{p},\om}(\cdot,\Re)),\Om} + \\
 + \left\|\nabla\u, \nabla^2\u\right\|_{p,\mu_{1-\frac{2}{p}}^{\frac{3}{2}-\frac{3}{p},\om}(\cdot,\Re)),\Om} + \left\|q,\nabla q\right\|_{p,\mu_{\frac{1}{2}-\frac{2}{p}}^{\frac{3}{2}-\frac{3}{p},\om}(\cdot,\Re)),\Om},
\end{split}
\end{equation}
where $p$ is sufficiently large and $\om < \frac{1}{2} - \frac{3}{2p}$.

\begin{rem}\label{weightedSobolev}
There exist a constant $C$ depending only on $\Om,A,B,\om$ such that for $p > 3$, $A,B \geq 0$, $\om \in [0,A]$ and $\Re \leq 1$
\begin{equation}
\left\|g\right\|_{L^\infty(\Om,\mu^{A,\om}_B(\cdot,\Re))} \leq C(\left\|g\right\|_{L^p(\Om,\mu^{A,\om}_B(\cdot,\Re))} + \left\|\nabla g\right\|_{L^p(\Om,\mu^{A,\om}_B(\cdot,\Re))}).
\end{equation}
This is an easy consequence of the Sobolev imbedding theorem and the fact that there is a constant $C$ independent of $\Re$ such that
\begin{equation}
\left\|g\right\|_{L^p(\Om,\nabla\mu^{A,\om}_B(\cdot,\Re))} \leq C\left\|g\right\|_{L^p(\Om,\mu^{A,\om}_B(\cdot,\Re))}.
\end{equation}
\end{rem}

Let us mention that using this Remark we get $\u \in L^\infty(\Om,\mu_{1-\frac{2}{p}}^{1-\frac{3}{p},\om}(\cdot,\Re))$ with $p$ arbitrarily large and therefore almost the same asymptotic structure as $\bO$.

Let us assume
\begin{equation}
\left\|(\w,s)\right\|_V \leq C_0,
\end{equation}
where $C_0$ is sufficiently large constant which will be determined later. It is important to mention that this constant is determined by the estimates (\ref{ass1}) and is independent of $\Re$ and $\We$. 

Our aim is to prove that also 
\begin{equation}
\left\|\M(\w,s)\right\|_V = \left\|(\u,q)\right\| \leq C_0.
\end{equation}
We recall that we also assume that $(\w,s)$ satisfy (\ref{ass1}).

Throughout the rest of this paper we will use the following notation to simplify things
\begin{equation}
X^{\om} = L^p(\Om,\mu_{1-\frac{2}{p}}^{\frac{3}{2}-\frac{3}{p},\om}(\cdot,\Re))
\end{equation}

We will estimate both $\B(\f,\w,s)$ and $\C(\H,\w,s)$ in $X^{2\om}$. Due to the presence of the Reynolds and Weissenberg numbers in front of each term on the right hand side it is sufficient to show the presence of all terms in $X^{2\om}$, smallness of these terms is achieved by assuming $\Re, \We$ sufficiently small. We will proceed term by term and denote the terms on the right hand side of (\ref{transport}) by $B_1, ..., B_6$ and the corresponding terms of $\C$ by $C_1, ..., C_6$. First we use assumption
\begin{equation}
\f,\H \in X^{2\om},
\end{equation}
which allows us to estimate $B_1, C_1$ in $X^{2\om}$. We estimate $B_2$ in the following way
\begin{equation}
\left\|\w\nabla\w\right\|_{X^{2\om}}^p \leq \left\|\w\right\|_{L^\infty(\Om,\mu_{1-\frac{2}{p}}^{1-\frac{3}{p},\om}(\cdot,\Re))}^p\left\|\nabla\w\right\|_{X^{\om}}^p\left\|\eta^{3-p}_{2-p}(\cdot,\Re)\right\|_{L^\infty(\Om)}.
\end{equation}
The last term is finite for $p > 3$ and therefore using Remark \ref{weightedSobolev}
\begin{equation}
\left\|B_2\right\|_{X^{2\om}} \leq C\Re C_0^2.
\end{equation}
We proceed in the similar way also in the divergence form case. Here
\begin{equation}
\left\|\w\otimes\w\right\|_{X^{2\om}}^p \leq \left\|\w\right\|_{L^\infty(\Om,\mu_{1-\frac{2}{p}}^{1-\frac{3}{p},\om}(\cdot,\Re))}^p\left\|\w\right\|_{L^p(\Om,\mu_{1-\frac{2}{p}}^{1-\frac{3}{p},\om}(\cdot,\Re))}^p\left\|\eta^{3-\frac{p}{2}}_{2-p}(\cdot,\Re)\right\|_{L^\infty(\Om)}
\end{equation}
and therefore for $p > 6$ we get
\begin{equation}
\left\|C_2\right\|_{X^{2\om}} \leq C\Re C_0^2.
\end{equation}
Similar procedure works also for terms $B_3$, $B_4$, $B_5$ and $C_3$, $C_4$, $C_5$, we only show the estimates for $B_3$ and $B_4$ which are most restritive.
\begin{equation}
\left\|\nabla^2\w\nabla\w\right\|_{X^{2\om}}^p \leq \left\|\nabla\w\right\|_{L^\infty(\Om,\mu_{1-\frac{2}{p}}^{\frac{3}{2}-\frac{3}{p},\om}(\cdot,\Re))}^p\left\|\nabla^2\w\right\|_{L^p(\Om,\mu_{1-\frac{2}{p}}^{\frac{3}{2}-\frac{3}{p},\om}(\cdot,\Re))}^p\left\|\eta^{3-\frac{3p}{2}}_{2-p}(\cdot,\Re)\right\|_{L^\infty(\Om)}
\end{equation}
\begin{equation}
\left\|\nabla\w\nabla s\right\|_{X^{2\om}}^p \leq \left\|\nabla\w\right\|_{L^\infty(\Om,\mu_{1-\frac{2}{p}}^{\frac{3}{2}-\frac{3}{p},\om}(\cdot,\Re))}^p\left\|\nabla s\right\|_{L^p(\Om,\mu_{\frac{1}{2}-\frac{2}{p}}^{\frac{3}{2}-\frac{3}{p},\om}(\cdot,\Re))}^p\left\|\eta^{3-\frac{3p}{2}}_{2-\frac{p}{2}}(\cdot,\Re)\right\|_{L^\infty(\Om)}
\end{equation}
Linear terms $B_6$, $C_6$ are trivial. We get 
\begin{equation}
\begin{split}
\left\|B_3,B_4,C_3,C_4\right\|_{X^{2\om}} \leq C\We C_0^2, \\
\left\|B_5,C_5\right\|_{X^{2\om}} \leq C\Re\We C_0^2, \\
\left\|B_6,C_6\right\|_{X^{2\om}} \leq \Re^{1-\om}\We C_0 
\end{split}
\end{equation}
and putting all calculations together we end up with
\begin{equation}
\begin{split}
\left\|\B(\f,\w,s)\right\|_{X^{2\om}} \leq C((\Re+\We)C_0^2 + \Re^{1-\om}\We C_0), \\
\left\|\C(\H,\w,s)\right\|_{X^{2\om}} \leq C((\Re+\We)C_0^2 + \Re^{1-\om}\We C_0),
\end{split}
\end{equation}
for $p > 6$. Now we can use Theorem \ref{weightedtrans} on the equation (\ref{transport}) to get
\begin{equation}
\left\|\z\right\|_{X^{2\om}} \leq C\left\|\B(\f,\w,s)\right\|_{X^{2\om}} \leq C((\Re+\We)C_0^2 + \Re^{1-\om}\We C_0).
\end{equation}
Moreover we can write the equation (\ref{transport}) in the following form
\begin{equation}
\z = \div\left[\C(\H,\w,s) - \We\z\otimes(\w + \e_1)\right]
\end{equation}
since $\div(\w + \e_1) = 0$. Hence $\z = \div \Z$ for some tensor $\Z$ and 
\begin{equation}
\left\|\Z\right\|_{X^{2\om}} \leq \left\|\C(\H,\w,s)\right\|_{X^{2\om}} + \We\left\|\z\right\|_{X^{2\om}}\left\|(\w + \e_1)\right\|_\infty \leq C((\Re+\We)C_0^2 + \Re^{1-\om}\We C_0).
\end{equation}

Now we can proceed with the Oseen equation (\ref{Oseen}) with the right hand side $\z = \div \Z$. We use the integral representation formulas (\ref{int0}) - (\ref{int2}) and (\ref{intp0}) - (\ref{intp1}) to estimate $(\u,q)$ in $V$. We can split $\u$ into $\u = \u^V + \u^S$, where $\u^V$ denotes the volume integral and $\u^S$ denotes the surface integrals. Similarly we split $\nabla\u$, $\nabla^2\u$, $q$ and $\nabla q$.

We start with the estimates of the volume parts. For the estimate of $\u^V \in L^p(\Om,\mu_{1-\frac{2}{p}}^{1-\frac{3}{p},\om}(\cdot,\Re))$ we use Theorem \ref{theo6} and its Corollary and get 
\begin{equation}\label{uV}
\left\|\u^V\right\|_{L^p(\mu_{1-\frac{2}{p}}^{1-\frac{3}{p},\om}(\cdot,\Re))} \leq C\Re^{-1+\om}\left\|\Z\right\|_{X^{2\om}} \leq C((\Re^\om + \Re^{-1+\om}\We)C_0^2 + \We C_0),
\end{equation}
which can be made sufficiently small by choosing $\Re, \We$ small.

For the estimates of $(\nabla\u)^V$ and $(\nabla^2\u)^V$ we use Theorem \ref{Koch} and its Corollary and get 
\begin{equation}
\left\|(\nabla\u)^V\right\|_{L^p(\mu_{1-\frac{2}{p}}^{\frac{3}{2}-\frac{3}{p},\om}(\cdot,\Re))} \leq C\Re^{\om}\left\|\Z\right\|_{X^{2\om}} \leq C((\Re^{1+\om}+\Re^\om\We)C_0^2 + \Re\We C_0)
\end{equation}
\begin{equation}
\left\|(\nabla^2\u)^V\right\|_{L^p(\mu_{1-\frac{2}{p}}^{\frac{3}{2}-\frac{3}{p},\om}(\cdot,\Re))} \leq C\Re^{\om}\left\|\z\right\|_{X^{2\om}} \leq C((\Re^{1+\om}+\Re^\om\We)C_0^2 + \Re\We C_0)
\end{equation}
Again, terms on the right-hand sides can be made sufficiently small by choosing $\Re, \We$ small. 
For the estimates of $q^V$ and $(\nabla q)^V$ we use Theorem \ref{theo7} and its Corollary and we obtain
\begin{multline}
\left\|q^V\right\|_{L^p(\mu_{\frac{1}{2}-\frac{2}{p}}^{\frac{3}{2}-\frac{3}{p},\om}(\cdot,\Re))} \leq C\left\|\Z\right\|_{L^p(\mu_{\frac{1}{2}-\frac{2}{p}}^{\frac{3}{2}-\frac{3}{p},\om}(\cdot,\Re))} \leq \\ 
\leq C\left\|\Z\right\|_{X^{2\om}} \leq C((\Re+\We)C_0^2 + \Re^{1-\om}\We C_0)
\end{multline}
\begin{multline}
\left\|(\nabla q)^V\right\|_{L^p(\mu_{\frac{1}{2}-\frac{2}{p}}^{\frac{3}{2}-\frac{3}{p},\om}(\cdot,\Re))} \leq  C\left\|\z\right\|_{L^p(\mu_{\frac{1}{2}-\frac{2}{p}}^{\frac{3}{2}-\frac{3}{p},\om}(\cdot,\Re))} \leq \\
 \leq C\left\|\z\right\|_{X^{2\om}} \leq C((\Re+\We)C_0^2 + \Re^{1-\om}\We C_0)
\end{multline}
Again the right-hand sides can be made small same way as before.

Next we proceed with the surface integrals. Here we distinguish three cases 
\begin{equation} 
\begin{split}
\Om_1 &= \left\{\x \in \Om, \abs{\x} \leq 1\right\} \\
\Om_2 &= \left\{\x \in \Om, 1 \leq \abs{\x} \leq \frac{1}{\Re}\right\} \\
\Om_3 &= \left\{\x \in \Om, \abs{\x} \geq \frac{1}{\Re}\right\}.
\end{split}
\end{equation}
In the case $\Om_1$ all our weights $\sim 1$ and we do not use the integral representation formulas. We rather use the following estimate
\begin{equation}
\left\|\u\right\|_{L^p(\Om_1)} \leq C(1 + \left\|\nabla\u\right\|_{W^{1,2}(\Om_1)}) \leq C(1 + \left\|\nabla\u\right\|_{W^{1,2}(\Om)}) \leq C(1+K)
\end{equation}
which is due to Friedrichs inequality and (\ref{ass1}). Arising term $C(1+K)$ can be made small in comparison with $C_0$ by choosing $C_0$ large enough. Together with (\ref{uV}) we get 
\begin{equation}
\left\|\u^S\right\|_{L^p(\Om_1,\mu_{1-\frac{2}{p}}^{1-\frac{3}{p},\om}(\cdot,\Re))} \leq C(1+K) + C((\Re^\om + \Re^{-1+\om}\We)C_0^2 + \We C_0).
\end{equation}
We use analogous procedure also for $\nabla\u, \nabla^2\u$ and get 
\begin{equation}
\left\|\nabla\u,\nabla^2\u\right\|_{L^p(\Om_1)} \leq C(1 + \left\|\nabla\u\right\|_{W^{3,2}(\Om_1)}) \leq C(1 + \left\|\nabla\u\right\|_{W^{3,2}(\Om)}) \leq C(1+K)
\end{equation}
and therefore
\begin{equation}
\left\|(\nabla\u)^S, (\nabla^2\u)^S\right\|_{L^p(\Om_1,\mu_{1-\frac{2}{p}}^{\frac{3}{2}-\frac{3}{p},\om}(\cdot,\Re))} \leq C(1+K) + C((\Re^{1+\om}+\Re^\om\We)C_0^2 + \Re\We C_0).
\end{equation}
Analogously for the pressure.

Next we proceed with the case $\Om_2$. We start with $\u^S$ and denote four terms in the surface integral (\ref{int0}) by $\u^{S,1}$, ..., $\u^{S,4}$. For $\u^{S,1}$ we have
\begin{multline}
\left|\u^{S,1}(\x)\right|^p\abs{\x}^{p\om}(1+\Re\abs{\x})^{p-3-p\om}(1+s(\Re\x))^{p-2} \leq \\ \leq \Re^p\abs{\x}^{p\om}(1+\Re\abs{\x})^{p-3-p\om}(1+s(\Re\x))^{p-2}\left|\int_{\de\Om}\O_{ij}(\x-\y,\Re)\d\y\right|^p \leq \\
\leq C\Re^p\abs{\x}^{p\om}(1+\Re\abs{\x})^{p-3-p\om}(1+s(\Re\x))^{p-2}\left|\O_{ij}(\x,\Re) + \nabla\O_{ij}(\frac{\x}{2},\Re)\right|^p \leq \\
\leq C\Re^p(1+\Re\abs{\x})^{p-3-p\om}(1+s(\Re\x))^{p-2}\left(\frac{1}{\abs{\x}^{p-p\om}} + \frac{1}{\abs{\x}^{2p-p\om}}\right),
\end{multline} 
where the crucial estimate is the following
\begin{equation}
\left|\nabla^k\bO(\x,\Re)\right| \leq C\frac{\Re^{\frac{k}{2}}}{\abs{\x}^{1+\frac{k}{2}}}
\end{equation}
for $k \geq 0$. We use this estimate throughout the rest of the procedure in the case $\Om_2$. Therefore
\begin{equation}
\left\|\u^{S,1}\right\|_{L^p(\Om_2,\mu_{1 - \frac{2}{p}}^{1 - \frac{3}{p},\om}(\cdot,\Re))}^p \leq C\Re^p\int_{\Om_2}\left(\frac{1}{\abs{\x}^{p-p\om}} + \frac{1}{\abs{\x}^{2p-p\om}}\right)\d\x
\end{equation}
Arising functions are integrable and $\int_{\Om_2} \abs{\x}^{p\om-p} \d\x \leq \frac{4\pi}{p-p\om-3} \leq 4\pi$ for $p>6$ and $\om < \frac{1}{2}-\frac{3}{2p}$, i.e. integrals of such functions over $\Om_2$ are bounded independently of $\Re$ by universal constant $4\pi$. Therefore this term can be estimated choosing $C_0$ large and at this point we do not require $\Re$ to be small, even if we have $\Re^p$ at our disposal. This fact will play a role in estimating $\u^{S,2}$, $\u^{S,3}$ and $\u^{S,4}$, where there is no power of $\Re$ available.  

Next for $\u^{S,2}$ we proceed similarly
\begin{multline}
\left|\u^{S,2}(\x)\right|^p\abs{\x}^{p\om}(1+\Re\abs{\x})^{p-3-p\om}(1+s(\Re\x))^{p-2} \leq \\
\leq C\abs{\x}^{p\om}(1+\Re\abs{\x})^{p-3-p\om}(1+s(\Re\x))^{p-2} \times \\ \times \left|\nabla\O_{ij}(\x,\Re) + \nabla^2\O_{ij}(\frac{\x}{2},\Re) + e_i(\x) + \nabla e_i(\frac{\x}{2})\right|^p \leq \\
\leq C(1+\Re\abs{\x})^{p-3-p\om}(1+s(\Re\x))^{p-2}\left(\frac{1}{\abs{\x}^{2p-p\om}} + \frac{1}{\abs{\x}^{3p-p\om}}\right)
\end{multline}
and therefore
\begin{equation}
\left\|\u^{S,2}\right\|_{L^p(\Om_2,\mu_{1 - \frac{2}{p}}^{1 - \frac{3}{p},\om}(\cdot,\Re))}^p \leq C\int_{\Om_2}\left(\frac{1}{\abs{\x}^{2p-p\om}} + \frac{1}{\abs{\x}^{3p-p\om}}\right)\d\x
\end{equation}
and we are in similar situation as in the case $\u^{S,1}$.

We treat $\u^{S,3}$ and $\u^{S,4}$ together
\begin{multline}
\left|\u^{S,3}+\u^{S,4}(\x)\right|^p\abs{\x}^{p\om}(1+\Re\abs{\x})^{p-3-p\om}(1+s(\Re\x))^{p-2} \leq \\
\leq C\abs{\x}^{p\om}(1+\Re\abs{\x})^{p-3-p\om}(1+s(\Re\x))^{p-2}\left|\O_{ij}(\x,\Re) + \nabla\O_{ij}(\frac{\x}{2},\Re)\right|^p \times \\
\times\left(\left\|\nabla\u\right\|_{W^{1,2}(\Om)} + \left\|q\right\|_{W^{1,2}(\Om)} + \left\|\Z\right\|_{W^{1,2}(\Om)}\right) \leq \\
\leq CK(1+\Re\abs{\x})^{p-3-p\om}(1+s(\Re\x))^{p-2}\left(\frac{1}{\abs{\x}^{p-p\om}} + \frac{1}{\abs{\x}^{2p-p\om}}\right)
\end{multline}
and therefore
\begin{equation}
\left\|\u^{S,3} + \u^{S,4}\right\|_{L^p(\Om_2,\mu_{1 - \frac{2}{p}}^{1 - \frac{3}{p},\om}(\cdot,\Re))}^p \leq CK\int_{\Om_2}\left(\frac{1}{\abs{\x}^{p-p\om}} + \frac{1}{\abs{\x}^{2p-p\om}}\right)\d\x.
\end{equation}
Here we have used also (\ref{ass1}).

For higher gradients of $\u$ and pressure and its gradient we use similar procedure, in this case higher gradients are even easier to estimate.

We finish with the case $\Om_3$. Here the situation is a little different. We have
\begin{multline}
\left|\u^{S,1}(\x)\right|^p\abs{\x}^{p\om}(1+\Re\abs{\x})^{p-3-p\om}(1+s(\Re\x))^{p-2} \leq \\ \leq \Re^p\abs{\x}^{p\om}(1+\Re\abs{\x})^{p-3-p\om}(1+s(\Re\x))^{p-2}\left|\int_{\de\Om}\O_{ij}(\x-\y,\Re)\d\y\right|^p \leq \\
\leq C\Re^p\abs{\x}^{p\om}(1+\Re\abs{\x})^{p-3-p\om}(1+s(\Re\x))^{p-2}\left|\O_{ij}(\x,\Re) + \nabla\O_{ij}(\frac{\x}{2},\Re)\right|^p \leq \\
\leq C\Re^p\abs{\x}^{p\om}(1+\Re\abs{\x})^{p-3-p\om}(1+s(\Re\x))^{p-2} \times \\ \times\left(\frac{\Re^p}{\abs{\Re\x}^p(1+s(\Re\x))^p} + \frac{\Re^{2p}}{\abs{\Re\x}^{\frac{3p}{2}}(1+s(\Re\x))^{\frac{3p}{2}}}\right).
\end{multline}
Here we have used that
\begin{equation}
\left|\nabla^k\bO(\x,\Re)\right| \leq C\frac{\Re^{\frac{k}{2}}}{\abs{\x}^{1+\frac{k}{2}}(1+s(\Re\x))^{1+\frac{k}{2}}}
\end{equation}
for $k \geq 0$. Therefore
\begin{multline}
\left\|\u^{S,1}\right\|_{L^p(\Om_3,\mu_{1 - \frac{2}{p}}^{1 - \frac{3}{p},\om}(\cdot,\Re))}^p \leq C \Re^{2p-p\om} \int_{\Om_3} \left(\frac{1}{(1+\abs{\Re\x})^{3}}\frac{1}{(1+s(\Re\x))^{2}}\right)\d\x + \\
+ C \Re^{3p-p\om} \int_{\Om_3} \left(\frac{1}{(1+\abs{\Re\x})^{3+\frac{p}{2}}}\frac{1}{(1+s(\Re\x))^{2+\frac{p}{2}}}\right)\d\x \leq \\ \leq C\Re^{2p-p\om-3}\int_{\rtri}\eta^{-3}_{-2}(\y)\d\y + C\Re^{3p-p\om-3}\int_{\rtri}\eta^{-3-\frac{p}{2}}_{-2-\frac{p}{2}}(\y)\d\y
\end{multline}
Arising integrals are finite due to (\ref{integrationeta}).

For $\u^{S,2}$ we obtain in the similar way
\begin{multline}
\left\|\u^{S,2}\right\|_{L^p(\Om_3,\mu_{1 - \frac{2}{p}}^{1 - \frac{3}{p},\om}(\cdot,\Re))}^p \leq C\Re^{2p-p\om}\int_{\Om_3}\left(\frac{1}{(1+\abs{\Re\x})^{3+\frac{p}{2}}}\frac{1}{(1+s(\Re\x))^{2+\frac{p}{2}}}\right)\d\x + \\
+ C\Re^{2p-p\om}\int_{\Om_3}\left(\frac{1}{(1+\abs{\Re\x})^{3+p}}\frac{1}{(1+s(\Re\x))^{2-p}}\right)\d\x + \\
+ C\Re^{3p-p\om}\int_{\Om_3}\left(\frac{1}{(1+\abs{\Re\x})^{3+p}}\frac{1}{(1+s(\Re\x))^{2+p}}\right)\d\x + \\
+ C\Re^{3p-p\om}\int_{\Om_3}\left(\frac{1}{(1+\abs{\Re\x})^{3+2p}}\frac{1}{(1+s(\Re\x))^{2-p}}\right)\d\x.
\end{multline}
Treating $\u^{S,3}$ and $\u^{S,4}$ together we get 
\begin{multline}
\left\|\u^{S,3} + \u^{S,4}\right\|_{L^p(\Om_3,\mu_{1 - \frac{2}{p}}^{1 - \frac{3}{p},\om}(\cdot,\Re))}^p \leq C\Re^{p-p\om}\int_{\Om_3}\left(\frac{1}{(1+\abs{\Re\x})^{3}}\frac{1}{(1+s(\Re\x))^{2}}\right)\d\x + \\
+ C\Re^{2p-p\om}\int_{\Om_3}\left(\frac{1}{(1+\abs{\Re\x})^{3+\frac{p}{2}}}\frac{1}{(1+s(\Re\x))^{2+\frac{p}{2}}}\right)\d\x \leq \\ \leq C\Re^{p-p\om-3}\int_{\rtri}\eta^{-3}_{-2}(\y)\d\y + C\Re^{2p-p\om-3}\int_{\rtri}\eta^{-3-\frac{p}{2}}_{-2-\frac{p}{2}}(\y)\d\y.
\end{multline}

Again, we proceed similarly with the estimates of gradients of $\u$ and pressure and its gradient. Putting all calculations together, choosing first $C_0$ sufficiently large and then $\Re, \We$ sufficiently small we finally end up with
\begin{equation}
\left\|(\u,q)\right\|_V < C_0
\end{equation}
and the proof of Theorem \ref{maintheo} is finished.

\end{document}